\documentclass{article}
\usepackage{dcase2024}
\usepackage{siunitx}
\usepackage{amssymb,amsfonts,amsthm,graphicx,mathtools}
\usepackage{cite}
\usepackage{hyperref}

\usepackage{adjustbox}
\usepackage{booktabs}
\usepackage{multirow}
\usepackage{acronym}
\usepackage{tikz,pgfplots}
\usepackage{csquotes}
\pgfplotsset{compat=1.18}
\usetikzlibrary{shapes.geometric, arrows, calc, fit, shapes, backgrounds}

\theoremstyle{plain}
\newtheorem{thm}{Theorem}
\newtheorem{lem}[thm]{Lemma}

\theoremstyle{definition}
\newtheorem{defn}[thm]{Definition}

\theoremstyle{remark}
\newtheorem*{rem}{Remark}

\acrodef{asd}[ASD]{anomalous sound detection}
\acrodef{gmm}[GMM]{Gaussian mixture model}
\acrodef{dg}[DG]{domain generalization}
\acrodef{auc}[AUC]{area under the \ac{roc} curve}
\acrodef{pca}[PCA]{principal component analysis}
\acrodef{lda}[LDA]{linear discriminant analysis}
\acrodef{lime}[LIME]{local interpretable model-agnostic explanations}
\acrodef{slime}[SLIME]{sound \ac{lime}}
\acrodef{dnn}[DNN]{deep neural network}
\acrodef{cd}[CD]{cosine distance}
\acrodef{ln}[LN]{length normalization}
\acrodef{rise}[RISE]{randomized input sampling for explanation}
\acrodef{roc}[ROC]{receiver operating characteristic}
\acrodef{tsne}[t-SNE]{t-distributed stochastic neighbor embedding}
\acrodef{xai}[xAI]{explainable artificial intelligence}
\acrodef{pauc}[pAUC]{partial area under the \ac{roc} curve}
\acrodef{umap}[UMAP]{uniform manifold approximation and projection}
\acrodef{oe}[OE]{outlier exposure}
\acrodef{im}[IM]{inlier modeling}
\acrodef{ic}[IC]{intra-class}
\acrodef{cce}[CXE]{categorical cross-entropy}
\acrodef{ssl}[SSL]{self-supervised learning}

\newcommand{\BE}{\begin{equation*}\begin{aligned}}
\newcommand{\EE}{\end{aligned}\end{equation*}}

\DeclareMathOperator{\Span}{span}

\title{AdaProj: Adaptively Scaled Angular Margin Subspace Projections for Anomalous Sound Detection with Auxiliary Classification Tasks}
\name{Kevin Wilkinghoff}
\address{Fraunhofer FKIE, Fraunhoferstr. 20 53353 Wachtberg, Germany\\
kevin.wilkinghoff@ieee.org}
\begin{document}
\ninept
\maketitle
\begin{sloppy}
\begin{abstract}
The state-of-the-art approach for semi-supervised anomalous sound detection is to first learn an embedding space by using auxiliary classification tasks based on meta information or self-supervised learning and then estimate the distribution of normal data.
In this work, AdaProj a novel loss function for training the embedding model is presented.
In contrast to commonly used angular margin losses, which project data of each class as close as possible to their corresponding class centers, AdaProj learns to project data onto class-specific subspaces while still ensuring an angular margin between classes.
By doing so, the resulting distributions of the embeddings belonging to normal data are not required to be as restrictive as other loss functions allowing a more detailed view on the data.
In experiments conducted on the DCASE2022 and DCASE2023 anomalous sound detection datasets, it is shown that using AdaProj to learn an embedding space significantly outperforms other commonly used loss functions.
\end{abstract}
\section{Introduction}
Semi-supervised anomaly detection is the task of training a system to differentiate between normal and anomalous data using only normal training samples \cite{aggarwal2017outlier}.
An example application is acoustic machine condition monitoring for predictive maintenance \cite{dohi2022description,dohi2023description}.
Here, normal data corresponds to sounds of fully functioning machines whereas anomalous sounds indicate mechanical failure.
One of the main difficulties to overcome in acoustic machine condition monitoring is that it is practically impossible to record isolated sounds of a target machine.
Instead, recordings also contain many other sounds emitted by non-target machines or other sound sources such as humans.
Compared to this complex acoustic scene, anomalous signal components of the target machines are very subtle and hard to detect without utilizing additional knowledge.
Another main difficulty is that a system should also be able to reliably detect anomalous sounds when changing the acoustic conditions or machine settings without needing to collect large amounts of data in these changed conditions or to re-train the system (domain generalization \cite{wang2021generalizing}).
One possibility to overcome both difficulties is to learn a mapping of the audio signals into a fixed-dimensional vector space, in which representations belonging to normal and anomalous data, called embeddings, can be easily separated.
Then, by estimating the distribution of normal training samples in the embedding space, one can compute an anomaly score for a test sample to distinguish between normal and anomalous samples.
\par
To train such an embedding model, the state-of-the-art is to utilize an auxiliary classification task using provided meta information or \acl{ssl}.
This enables the embedding model to closely monitor target signals and ignore other signals and noise \cite{wilkinghoff2024why}.
For machine condition monitoring, possible auxiliary tasks are classifying between machine types \cite{giri2020self,lopez2020speaker,inoue2020detection} or, additionally, between different machine states and noise settings \cite{venkatesh2022improved,nishida2022anomalous,deng2022ensemble}, recognizing augmented and non-augmented versions of normal data \cite{giri2020self,chen2023effective} or predicting the activity of machines \cite{nishida2022anomalous}.
Using an auxiliary task to learn embeddings is also called \ac{oe} \cite{hendrycks2019deep} because normal samples belonging to other classes than a target class can be considered proxy outliers \cite{primus2020anomalous}.
\par
The contributions of this work are the following.
First and foremost, AdaProj, a novel angular margin loss function that learns class-specific subspaces for training an embedding model, is presented\footnote{An open-source implementation of the AdaProj loss is available at: \url{https://github.com/wilkinghoff/AdaProj}}.
Second, it is proven that AdaProj has arbitrary large optimal solution spaces allowing to relax the compactness requirements of the class-specific distributions in the embedding space.
Last but not least, AdaProj is compared to other commonly used loss functions.
In experiments conducted on the DCASE2022 and DCASE2023 \ac{asd} datasets it is shown that AdaProj outperforms other commonly used loss functions.

\subsection{Related Work}
When training a neural network to solve a classification task, usually the softmax function in combination with the \ac{cce} is used.
However, this only reduces inter-class similarity without explicitly increasing intra-class similarity \cite{wang2018additive}.
When training an embedding model for anomaly detection, high intra-class similarity is a desired property to cluster normal data and be able to detect anomalous samples.
There are several loss functions that explicitly increase intra-class similarity:
\cite{ruff2018deep} proposed a compactness loss to project the data into a hypersphere of minimal volume for one-class classification.
However, for machine condition monitoring in noisy conditions it is known that one-class losses perform worse than losses that also discriminatively solve an auxiliary classification task \cite{wilkinghoff2024why}.
\cite{perera2019learning} utilized an additional descriptiveness loss consisting of a \ac{cce} imposing a classification task on another arbitrary dataset than the target dataset to regularize the training objective.
For machine condition monitoring, often meta information is available as it can at least be ensured which machine is being recorded when collecting data.
\cite{inoue2020detection} used center loss \cite{wen2016discriminative}, which minimizes the distance to learned class centers for each class.
Another choice are angular margin losses that learn an embedding space on the unit sphere while ensuring a margin between classes, which improves the generalization capabilities.
Specific examples are the additive margin softmax loss \cite{wang2018additive} as used by \cite{lopez2020speaker,lopez2021ensemble} and ArcFace \cite{deng2019arcface} as used by \cite{giri2020self,kuroyanagi2021ensemble,deng2022ensemble}. \cite{wilkinghoff2021combining,wilkinghoff2023design} use the AdaCos loss \cite{zhang2019adacos}, which essentially is ArcFace with an adaptive scale parameter, or the sub-cluster AdaCos loss \cite{wilkinghoff2021sub}, which utilizes multiple sub-clusters for each class instead of a single one.
\par
As stated before, the goal of this work is to reduce the restrictions on the learned distributions in the embedding space by learning class-specific linear subspaces.
There are also other works on losses aiming at learning subspaces based on orthogonal projections in an embedding space.
\cite{yu2021autoencoder} used orthogonal projections as a constraint for training an autoencoder based anomaly detection system.
Another example is semi-supervised image classification by using a combination of class-specific subspace projections with a reconstructions loss and ensure that they are different by also using a discriminative loss \cite{li2022learnable}.
Our work focuses on learning an embedding space through an auxiliary classification task that is well-suited for semi-supervised anomaly detection.

\section{Methodology}
\subsection{Notation}
Let $\phi:X\rightarrow\mathbb{R}^D$ denote a neural network where $X$ denotes some input space, which consists of audio signals in this work, and $D\in\mathbb{N}$ denotes the dimension of the embedding space.
Define the linear projection of $x\in\mathbb{R}^D$ onto the subspace $\Span(\mathcal{C}_k)\subset\mathbb{R}^D$ as ${P_{\Span(\mathcal{C}_k)}(x):=\sum_{c_k\in \mathcal{C}_k}\langle x,c_k\rangle c_k}$.
Furthermore, let ${\mathcal{S}^{D-1}=\lbrace y\in\mathbb{R}^D:\lVert y\rVert_2=1\rbrace\subset\mathbb{R}^D}$ denote the $D$-sphere and define ${P_{\mathcal{S}^{D-1}}(x):=\frac{x}{\lVert x\rVert_2}\in \mathcal{S}^{D-1}}$ to be the projection onto the $D$-sphere.

\subsection{AdaProj loss function}
Similar to the sub-cluster AdaCos loss \cite{wilkinghoff2021sub}, the idea of the AdaProj loss is to enlarge the space of optimal solutions to allow the network to learn less restrictive distributions of normal data.
This relaxation is achieved by measuring the distance to class-specific subspaces while training the embedding model instead of measuring the distance to a single or multiple centers as done for other angular margin losses and may help to differentiate between normal and anomalous embeddings after training.
The reason is that for some auxiliary classes a strong compactness may be detrimental when aiming to learn an embedding space that separates normal and anomalous data since both may be distributed very similarly.
\par
Formally, the definition of the AdaProj loss is as follows.
\begin{defn}[AdaProj loss]
Let $\mathcal{C}_k\subset\mathbb{R}^D$ with $\lvert\mathcal{C}_k\rvert=J\in\mathbb{N}$ denote class centers for class $k\in\lbrace1,...,N_\text{classes}\rbrace$.
Then for the AdaProj loss the logit for class $k\in\lbrace1,...,N_\text{classes}\rbrace$ is defined as
\BE L(x,\mathcal{C}_k):=\hat{s}\cdot\lVert P_{\mathcal{S}^{D-1}}(x)-P_{\mathcal{S}^{D-1}}(P_{{\Span(\mathcal{C}_k)}}(x))\rVert_2^2\EE
where $\hat{s}\in\mathbb{R}_+$ is the adaptive scale parameter of the AdaCos loss \cite{zhang2019adacos}.
Inserting these logits into a softmax function and computing the \ac{cce} yields the AdaProj loss function.
\end{defn}
\begin{rem}
Note that, by Lemma 5 of \cite{wilkinghoff2024why}, it holds that
\BE &\lVert P_{\mathcal{S}^{D-1}}(x)-P_{\mathcal{S}^{D-1}}(P_{{\Span(\mathcal{C}_k)}}(x))\rVert_2^2\\=&2(1-\langle P_{\mathcal{S}^{D-1}}(x),P_{\mathcal{S}^{D-1}}(P_{{\Span(\mathcal{C}_k)}}(x))\rangle),\EE
which is equal to the cosine distance in this case and explains why the AdaProj loss can be called an angular margin loss.
\end{rem}
As for other angular margin losses, projecting the embedding space onto the $D$-sphere has several advantages \cite{wilkinghoff2024why}.
Most importantly, if $D$ is sufficiently large randomly initialized centers are with very high probability approximately orthonormal to each other \cite{gorban2016approximation}, i.e. distributed equidistantly, and sufficiently far away from $\pmb{0}\in\mathbb{R}^D$.
Therefore, one does not need to carefully design a method to initialize the centers.
Another advantage is that a normalization may prevent numerical issues, similar to applying batch normalization \cite{ioffe2015batch}.
\par
The following Lemma shows that using the AdaProj loss, as defined above, indeed increases the solution space.
\begin{lem}
Let $x\in\mathbb{R}^D$ and let $\mathcal{C}\subset\mathbb{R}^D$ contain pairwise orthonormal elements. If $x\in\Span(\mathcal{C})\cap \mathcal{S}^{D-1}$, then
\BE\lVert P_{\mathcal{S}^{D-1}}(x)-P_{\mathcal{S}^{D-1}}(P_{\Span(\mathcal{C})}(x))\rVert_2^2=0.\EE
\end{lem}
\begin{proof}
    Let $x\in\Span(\mathcal{C})\cap \mathcal{S}^{D-1}\subset\mathbb{R}^D$ with $\lvert \mathcal{C}\rvert=J$. Therefore, $\lVert x\rVert_2=1$ and there are $\lambda_j\in\mathbb{R}$ with $x=\sum_{j=1}^J\lambda_jc_j$.
    Thus, it holds that
    \BE x&=\sum_{j=1}^J\lambda_jc_j=\sum_{j=1}^J\sum_{i=1}^J\lambda_i\langle c_i,c_j\rangle c_j=\sum_{j=1}^J\langle\sum_{i=1}^J\lambda_i c_i,c_j\rangle c_j\\&=\sum_{j=1}^J\langle x,c_j\rangle c_j=P_{\Span(\mathcal{C})}(x).\EE
    Hence, we obtain
    \BE &\lVert P_{\mathcal{S}^{D-1}}(x)-P_{\mathcal{S}^{D-1}}(P_{\Span(\mathcal{C})}(x))\rVert_2^2=&0.\EE
\end{proof}
\begin{rem}
    If $\mathcal{C}$ contains randomly initialized elements of the unit sphere and $D$ is sufficiently large, then the elements of $\mathcal{C}$ are approximately pairwise orthonormal with very high probability \cite{gorban2016approximation}.
\end{rem}
When inserting the projection onto the $D-1$-sphere as an operation into the neural network, this Lemma shows that the solution space for the AdaProj loss function is increased to the whole subspace $\Span(\mathcal{C})$, which has a dimension of $\lvert \mathcal{C}\rvert$ with very high probability.
Because of this, it should be ensured that $\lvert \mathcal{C}\rvert<D$.
Otherwise the whole embedding space may be an optimal solution and thus the network cannot learn a meaningful embedding space.
In comparison, for the AdaCos loss only the class centers themselves are optimal solutions and for the sub-cluster AdaCos loss each sub-cluster is an optimal solution \cite{wilkinghoff2024why}.

\section{Experimental results}

\subsection{Datasets and performance metrics}
For the experiments, the DCASE2022 \ac{asd} dataset \cite{dohi2022description} and the DCASE2023 \ac{asd} dataset \cite{dohi2023description} for semi-supervised machine condition monitoring were used.
Both datasets consist of a development set and an evaluation set that are divided into a training split containing only normal data and a test split containing normal as well as anomalous data.
Furthermore, both tasks explicitly capture the problem of domain generalization \cite{wang2021generalizing} by defining a source and a target domain, which differs from the source domain by altering machine parameters or noise conditions.
The task is to detect anomalous samples regardless of the domain a sample belongs to by training a system with only normal data.
As meta information, the target machine type of each sample is known and for the training samples, also the domain and additional parameter settings or noise conditions, called attribute information, are known and thus can be utilized to train an embedding model.
\par
The DCASE2022 \ac{asd} dataset \cite{dohi2022description} consists of the machine types \enquote{ToyCar} and \enquote{ToyTrain} from ToyAdmos2 \cite{harada2021toyadmos2} and \enquote{fan}, \enquote{gearbox}, \enquote{bearing}, \enquote{slide rail} and \enquote{valve} from MIMII-DG \cite{dohi2022mimii_dg}.
For each machine type, there are six different sections corresponding to different domain shifts and also defining subsets used for computing the performance.
These sections are known for each recording and can also be utilized as meta information to train the system.
For the source domain of each section, there are $1000$ normal audio recordings with a duration of \SI{10}{\second} and a sampling rate of \SI{16}{\kilo\hertz} belonging to the training split.
For the target domain of each section, there are only $10$ normal audio recordings belonging to the training split.
The test splits of each section contain approximately $100$ normal and $100$ anomalous samples.
\par
The DCASE2023 \ac{asd} dataset \cite{dohi2023description} is similar to the DCASE2022 \ac{asd} dataset with the following modifications.
First of all, the development set and the evaluation set contain mutually exclusive machine types.
More concretely, the development set contains the same machine types as the DCASE2022 dataset and the evaluation set contains the machine types
\enquote{ToyTank}, \enquote{ToyNscale} and \enquote{ToyDrone} from ToyAdmos2+ \cite{harada2023toyadmos2+} and \enquote{vacuum}, \enquote{bandsaw}, \enquote{grinder} and \enquote{shaker} from MIMII-DG \cite{dohi2022mimii_dg}.
Furthermore, there is only a single section for each machine type, which makes the auxiliary classification task much easier resulting in less informative embeddings for the \ac{asd} task.
Last but not least, the duration of each recording has a length between \SI{6}{\second} and \SI{18}{\second}.
Overall, all three modifications make this task much more challenging than the DCASE2022 \ac{asd} task.
\par
To measure the performance of the \ac{asd} systems the threshold-independent \ac{auc} metric is used.
In addition, the \ac{pauc} \cite{mcclish1989analyzing}, which is the \ac{auc} for low false positive rates ranging from $0$ to $p$, with $p=0.1$, is used.
Both performance metrics are computed domain-independent for every previously defined section of the dataset and the harmonic mean of all resulting values is the final score used to measure and compare the performances of different \ac{asd} systems.

\subsection{Anomalous sound detection system}
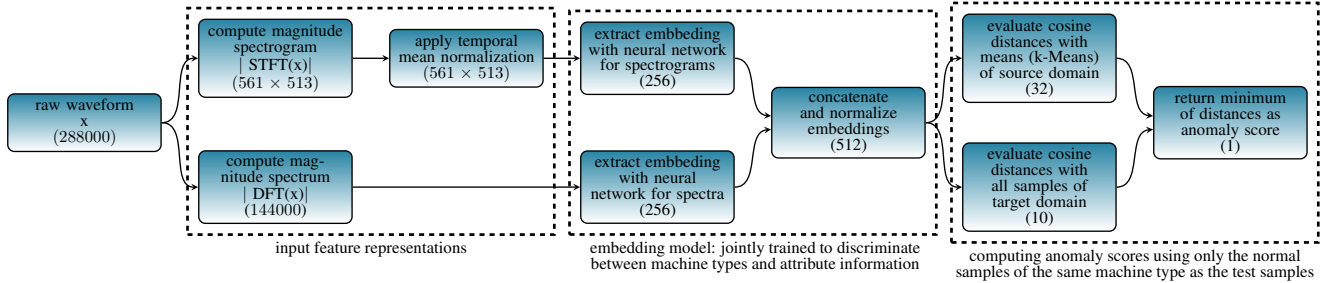
\begin{figure*}
	\centering
	\begin{adjustbox}{max width=\textwidth}
		\usetikzlibrary{shapes.geometric, arrows, calc, fit}
\tikzset{box/.style={top color=teal!85!blue!85,shade,draw, rectangle, rounded corners, thick, node distance=1cm, minimum width=2.3cm, text centered, minimum height=2em, text width=3cm}}
\tikzset{container/.style={draw, rectangle, dashed, ultra thick, inner sep=0.7em,minimum width=2cm}}
\tikzset{container2/.style={,draw, rectangle, dashed, ultra thick, inner sep=0.7em,minimum width=2cm}}
\tikzset{font={\fontsize{10pt}{10}\selectfont}}
\tikzstyle{arrow} = [thick, ->, >=stealth]
\begin{tikzpicture}
\node(wav)[box,node distance=1.65cm]{raw waveform\\ x\\ $(288000)$};
\node(log_mel)[box,right of=wav, yshift=1.35cm, node distance=4cm]{compute magnitude spectrogram\\ $\lvert$ STFT(x)$\rvert$\\ $(561\times513)$};
\node(cmn)[box,right of=log_mel, node distance=4cm]{apply temporal mean normalization $(561\times513)$};
\node(fft)[box,right of=wav, yshift=-1.35cm, node distance=4cm]{compute magnitude spectrum \\ $\lvert$ DFT(x)$\rvert$\\ ($144000$)};
\node(nnmel)[box,right of=cmn, node distance=4cm]{extract embbeding with neural network for spectrograms\\(256)};
\node(nnfft)[box,right of=fft, node distance=8cm]{extract embbeding with neural network for spectra\\(256)};
\node(concat)[box,right of=nnfft, yshift=1.35cm, node distance=4cm]{concatenate and normalize embeddings\\ (512)};
\node(gmmconcat1)[box,right of=concat, node distance=4cm, yshift=-1.35cm]{evaluate cosine distances with all samples of target domain\\(10)};
\node(gmmconcat2)[box,right of=concat, node distance=4cm, yshift=1.35cm]{evaluate cosine distances with means (k-Means) of source domain\\(32)};
\node(llh)[box,right of=gmmconcat2, node distance=4cm, yshift=-1.35cm]{return minimum of distances as anomaly score\\ (1)};
\begin{scope}[on background layer]
\node[container2, fit=(nnfft)(nnmel)(concat)] (alltrain) {};
\node at (alltrain.south west) [text centered,below right,node distance=0 and 0, align=center,xshift=0.35cm] (alltraintxt) {embedding model: jointly trained to discriminate\\ between machine types and attribute information};
\node[container2, fit=(fft)(log_mel)(cmn)] (frontend) {};
\node at (frontend.south) [text centered,below,node distance=0 and 0, align=center,xshift=0cm] (frontendtxt) {input feature representations};
\node[container2, fit=(gmmconcat1)(gmmconcat2)(llh)] (backend) {};
\node at (backend.south) [text centered,below,node distance=0 and 0, align=center,xshift=0cm] (backendtxt) {computing anomaly scores using only the normal\\ samples of the same machine type as the test samples};
\end{scope}

\draw[arrow](wav.0) [out=0, in=180] to (log_mel.180);
\draw[arrow](log_mel.0) [out=0, in=180] to (cmn.180);
\draw[arrow](cmn.0) [out=0, in=180] to (nnmel.180);
\draw[arrow](wav.0) [out=0, in=180] to (fft.180);
\draw[arrow](fft.0) [out=0, in=180] to (nnfft.180);
\draw[arrow](nnfft.0) [out=0, in=185] to (concat.185);
\draw[arrow](nnmel.0) [out=0, in=175] to (concat.175);
\draw[arrow](concat.0) [out=0, in=180] to (gmmconcat1.180);
\draw[arrow](concat.0) [out=0, in=180] to (gmmconcat2.180);
\draw[arrow](gmmconcat1.0) [out=0, in=185] to (llh.185);
\draw[arrow](gmmconcat2.0) [out=0, in=175] to (llh.175);

\end{tikzpicture}
	\end{adjustbox}
	\caption{Structure of the \ac{asd} system, adapted from Figure 1 in \cite{wilkinghoff2023design}. Representation size in each step is given in brackets.}
	\label{fig:system_structure}
\end{figure*}
For all experiments conducted in this work, the state-of-the-art \ac{asd}  system presented in \cite{wilkinghoff2023design} is used.
An overview of the system can be found in \autoref{fig:system_structure}.
The system consists of three main components: 1) a feature extractor, 2) an embedding model and 3) a backend for computing anomaly scores.
\par
In the first processing block, two different feature representations are extracted from the raw waveforms, namely magnitude spectrograms and the magnitude spectrum.
To capture less similar information with both feature representations, the temporal mean is subtracted from the spectrograms, essentially removing static frequency information that are captured with the highest possible resolution by the spectra.
\par
For each of the two feature representations, another convolutional sub-network is trained and the resulting embeddings are concatenated and normalized with respect to the Euclidean norm to obtain a single embedding.
In contrast to the original architecture, the embedding dimension is doubled from $256$ to $512$ to increase the likelihood of two randomly initialized center vectors to be orthogonal.
Note that  even if some of the randomly sampled class centers are not orthonal, the probability that they are linearly independent is equal to $1$ if $J<D$.
Thus, the subspaces spanned be the class centers do not collapse.
More details about the subnetwork architectures can be found in \cite{wilkinghoff2023design}.
The network is trained for $10$ epochs using a batch size of $64$ using adam \cite{kingma2015adam} by utilizing meta information such as machine types and the provided attribute information as an auxiliary classification task.
Different loss functions can be used for this purpose and will be compared in the next subsection.
All loss functions investigated in this work require class-specific center vectors, which are initialized randomly using Glorot uniform initialization \cite{glorot2010understanding}.
To improve the \ac{asd} performance, the class centers are not adapted during training and no bias terms are used as proposed in \cite{ruff2018deep} for one-class classification.
Furthermore, mixup \cite{zhang2017mixup} with a uniformly distributed mixing coefficient is applied to the waveforms.
\par
As a backend, k-means with $32$ means is applied to the normal training samples of the source domain.
For a given test sample, the smallest cosine distance to these means and the ten normal training samples of the target domain is used as an anomaly score.
Thus, smaller values indicate normal samples whereas higher values indicate anomalous samples.

\subsection{Performance evaluation}
\begin{table*}
	\centering
	\caption{\ac{asd} performance obtained with different loss functions. Harmonic means of all AUCs and pAUCs over all pre-defined sections of the dataset are depicted in percent. Arithmetic mean and standard deviation of the results over ten independent trials are shown. Best results in each column are highlighted with bold letters.}
\begin{adjustbox}{max width=\textwidth}
	\begin{tabular}{lll|ll|ll|ll|ll}
		\toprule
		\multirow{2}{*}{loss function}&\multicolumn{2}{l}{DCASE2022 dev. set \cite{dohi2022description}}&\multicolumn{2}{l}{DCASE2022 eval. set \cite{dohi2022description}}&\multicolumn{2}{l}{DCASE2023 dev. set \cite{dohi2023description}}&\multicolumn{2}{l}{DCASE2023 eval. set \cite{dohi2023description}}&\multicolumn{2}{l}{arithmetic mean}\\
		&AUC&pAUC&AUC&pAUC&AUC&pAUC&AUC&pAUC&AUC&pAUC\\
		\midrule
		\acs{ic} compactness loss \cite{ruff2018deep}&$79.2\pm0.9$&$64.7\pm1.1$&$70.3\pm0.8$&$58.9\pm0.8$&$67.7\pm1.2$&$56.9\pm0.9$&$64.0\pm1.5$&$55.8\pm0.9$&$70.3$&$59.1$\\
		\acs{ic} compactness loss + \acs{cce} \cite{perera2019learning}&$79.0\pm0.8$&$65.0\pm0.7$&$72.6\pm0.4$&$60.3\pm0.7$&$70.4\pm1.0$&$\pmb{57.4\pm1.1}$&$67.5\pm0.8$&$57.5\pm1.0$&$72.4$&$60.1$\\
		AdaCos loss \cite{zhang2019adacos}&$79.8\pm0.7$&$\pmb{65.5\pm0.9}$&$73.0\pm0.4$&$59.7\pm0.6$&$70.9\pm0.9$&$56.8\pm0.9$&$68.0\pm1.6$&$58.0\pm1.1$&$72.9$&$60.0$\\
		sub-cluster AdaCos loss \cite{wilkinghoff2021sub}&$80.0\pm1.4$&$65.2\pm1.1$&$72.9\pm0.6$&$59.5\pm0.5$&$70.4\pm0.9$&$56.3\pm0.8$&$66.5\pm1.6$&$56.2\pm1.0$&$72.5$&$59.3$\\
        proposed AdaProj loss&$\pmb{80.6\pm0.8}$&$\pmb{65.5\pm1.3}$&$\pmb{73.6\pm0.7}$&$\pmb{60.5\pm0.7}$&$\pmb{71.4\pm1.0}$&$56.2\pm0.7$&$\pmb{69.8\pm1.3}$&$\pmb{60.0\pm0.5}$&$\pmb{73.9}$&$\pmb{60.6}$\\
		\bottomrule
	\end{tabular}
\end{adjustbox}
\label{tab:performances}
\end{table*}
In the first experiment, the \ac{asd} performance obtained with the following loss functions was compared: 1) individual class-specific \ac{ic} compactness losses jointly trained on all classes \cite{ruff2018deep} 2) an additional discriminative \ac{cce} loss, similar to the descriptiveness loss used in \cite{perera2019learning} but trained on the same dataset, 3) the AdaCos loss \cite{zhang2019adacos}, 4) the sub-cluster AdaCos loss \cite{wilkinghoff2021sub} with $32$ sub-clusters and 5) the proposed AdaProj loss.
Each experiment was repeated ten times to reduce the variance of the resulting performances.
The results can be found in \autoref{tab:performances}.
\par
The main observation to be made is that the proposed AdaProj loss outperforms all other losses.
Especially on the DCASE2023 dataset, there are significant improvements to be observed.
The most likely explanation is that for this dataset the classification task is less difficult and thus a few classes may be easily identified leading to embeddings that do not carry enough information to distinguish between embeddings belonging to normal and anomalous samples of these classes.
\par
Another interesting observation is that, in contrast to the original results presented in \cite{wilkinghoff2021sub}, the sub-cluster AdaCos loss actually performs slightly worse than the AdaCos loss despite having a higher solution space.
A possible explanation is that in \cite{wilkinghoff2021sub}, the centers are adapted during training whereas, in our work, they are not as this has been shown to improve the resulting performance \cite{wilkinghoff2023design}.
Since all centers have approximately the same distance to each other when being randomly initialized, i.e. the centers belonging to a target class and the other centers, the network will likely utilize only a single center for each class that is closest to the initial embeddings of the corresponding target class.
Moreover, a low inter-class similarity is more difficult to ensure due to the higher total number of  sub-clusters belonging to other classes.
This leads to more restrictive requirements when learning class-specific distributions and thus actually reduces the ability to differentiate between embeddings belonging to normal and anomalous samples.

\subsection{Investigating the impact of the subspace dimension on the performance}
\begin{figure}
    \centering
    \begin{adjustbox}{max width=\columnwidth}
          \begin{tikzpicture}
\begin{axis}[
	axis y line*=left,
    axis x line*=bottom,
    ymin=0.5,
    ymax=0.75,
    xmin=1,
    xmax=128,
    enlarge x limits=0.05,
    legend style={at={(0.5,1.05)},anchor=south,legend columns=2},
    ylabel=harmonic mean of \acp{auc}/\acp{pauc},
    xlabel=subspace dimension $\lvert\mathcal{C}_k\rvert$,
    height=7cm,
    width=16cm,
    xticklabel style={align=center},
    yticklabel style={align=center},
    typeset ticklabels with strut,
    xlabel near ticks,
    ylabel near ticks,
    nodes near coords style={/pgf/number format/.cd,fixed zerofill,precision=2},
    ymajorgrids
]
{(1,0.7133)(2,0.7133)(3,0.7108)(4,0.7083)(5,0.7093)(6,0.7063)(7,0.7141)(8,0.7196)(9,0.7182)(10,0.7144)(11,0.7187)(12,0.7107)(13,0.7077)(14,0.7073)};
\addplot[teal!15!orange!85,mark=triangle,line width=2pt] coordinates
{(1,0.7133)(2,0.7133)(3,0.7108)(4,0.7083)(6,0.7093)(8,0.7063)(12,0.7141)(16,0.7196)(20,0.7101)(24,0.7182)(28,0.7115)(32,0.7144)(40,0.7127)(48,0.7187)(56,0.7088)(64,0.7107)(80,0.7030)(96,0.7077)(128,0.7073)};
\addplot[teal!15!orange!45,mark=diamond,line width=2pt] coordinates {(1,0.5653)(2,0.5701)(3,0.5601)(4,0.5645)(6,0.5663)(8,0.5575)(12,0.5675)(16,0.5674)(20,0.5639)(24,0.5632)(28,0.5652)(32,0.5615)(40,0.5624)(48,0.5699)(56,0.5625)(64,0.5676)(80,0.5575)(96,0.5648)(128,0.5555)};
\addplot[teal!85!blue!85,mark=square,line width=2pt] coordinates {(1,0.6896)(2,0.6886)(3,0.6874)(4,0.6952)(6,0.6909)(8,0.7025)(12,0.6918)(16,0.702)(20,0.6949)(24,0.6980)(28,0.7070)(32,0.6979)(40,0.7007)(48,0.7054)(56,0.7028)(64,0.7051)(80,0.7044)(96,0.6976)(128,0.6994)};
\addplot[teal!85!blue!45,mark=pentagon,line width=2pt] coordinates {(1,0.586)(2,0.5889)(3,0.5938)(4,0.5927)(6,0.5912)(8,0.6004)(12,0.5933)(16,0.5951)(20,0.5922)(24,0.6006)(28,0.6030)(32,0.5997)(40,0.5973)(48,0.6109)(56,0.6042)(64,0.6026)(80,0.6041)(96,0.6027)(128,0.6014)};
\legend{\ac{auc} - development set, \ac{pauc} - development set, \ac{auc} - evaluation set, \ac{pauc} - evaluation set}
\end{axis}
\end{tikzpicture}
    \end{adjustbox}
    \caption{Domain-independent performance obtained on the DCASE2023 dataset with different subspace dimensions. The means over ten independent trials are shown.}
    \label{fig:comparison}
\end{figure}
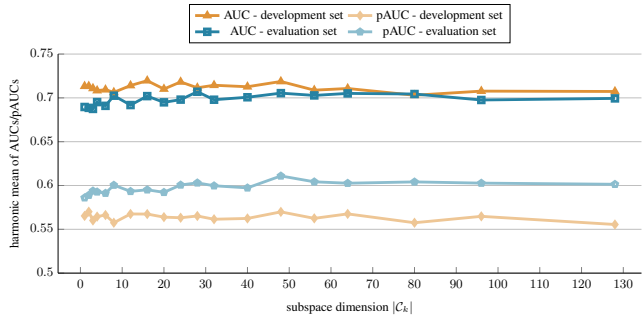
As an ablation study, different choices for the dimension of the subspaces have been compared experimentally on the DCASE2023 \ac{asd} dataset.
The results can be found in \autoref{fig:comparison}.
It can be seen, that, on the development set, the results are relatively stable while a larger dimension slightly improves the performance on the evaluation set without any significant differences.
For subspace dimensions greater than $48$ the performances seem to slightly degrade again.
In conclusion, the subspace dimension should be neither too high nor too low but other than that appears to not have a significant impact on the performance.
Thus, a dimension of $32$, as used for the other experiments in this work, appears to be a reasonable choice.
Since using the AdaProj with this subspace dimension also outperformed the other loss functions on the DCASE2022 \ac{asd} dataset (cf. \autoref{tab:performances}, this particular dimension may serve as a default hyperparameter setting for the AdaProj loss.

\section{Conclusions}
In this work, AdaProj a novel angular margin loss function specifically designed for semi-supervised anomaly detection with auxiliary classification tasks was presented.
It was proven that this loss function learns an embedding space with class-specific subspaces of arbitrary dimension.
In contrast to other angular margin losses, which try to project data to individual points in space, this relaxes the requirements of solving the classification task and allows for less compact distributions in the embedding space.
In experiments conducted on the DCASE2022 and DCASE2023 \ac{asd} datasets, it was shown that using AdaProj results in better performance than other commonly used loss functions.
In conclusion, the resulting embedding space has a more desirable structure than the other embedding spaces for differentiating between normal and anomalous samples.
For future work, it is planned to evaluate AdaProj on other datasets and further improve the performance of the \ac{asd} system by utilizing \acl{ssl} \cite{wilkinghoff2024self} or multi-task learning \cite{venkatesh2022improved}.
Investigating how the AdaProj loss performs for supervised or unsupervised tasks in comparison to other loss functions may also be of interest.

\bibliography{refs}
\bibliographystyle{IEEEbib}
\end{sloppy}
\end{document}